\documentclass[10pt,a4paper]{article}
\usepackage{fullpage}
\usepackage{amsfonts,amsmath,amssymb}
\usepackage{amsthm}
\usepackage{graphicx}\usepackage{sectsty}
\usepackage{authblk}

\theoremstyle{plain}
\newtheorem{theorem}{Theorem}[section]
\newtheorem{lemma}[theorem]{Lemma}

\theoremstyle{definition}
\newtheorem{definition}[theorem]{Definition}
\theoremstyle{remark}

\numberwithin{equation}{section}

\sectionfont{\large}
\newenvironment{acknowledgement}[1][Acknowledgement
]{\begin{trivlist} \item[\hskip \labelsep {\bfseries
#1}]}{\end{trivlist}}

\begin{document}
\title{Uniqueness of Inverse Spectral Problem of Non-Local Sturm-Liouville Operators on Star Graph}
\author[1]{Lung-Hui Chen}
\affil[1]{\footnotesize General Education Center, Ming Chi University of Technology, New Taipei City, 24301, Taiwan.}
\maketitle
\begin{abstract}
In this paper, we explore the inverse spectral problem of Sturm-Liouville operator on a star-like graph. To this fixed star-like graph centered at the origin as its vertex, we attach $m$ edges. On each edge, we impose the Sturm-Liouville operator with certain non-local potential functions with some suitable non-local boundary value conditions. At the vertex, we consider a frozen argument type of condition at zero to model a network that fixed on the end of each edge on the graph. The vibration and flow changes are monitored at that vertex which serves as certain control center.  There is an inverse uniqueness subject to the suitable non-local boundary condition. We show that the system is solvable. Additionally, we give a Weyl's type of spectral asymptotics.
\\MSC: 47A55/34A55/34K29.
\\Keywords:  Sturm-Liouville operator; star graph; inverse spectral problem; network theory; non-local regulator; interpolation theory; traffic control.
\end{abstract}
\section{Introduction}
In this paper, we discuss the mathematical theory of Sturm-Liouville operators on star graphs. It is one of the highly-pursued areas in modern mathematical physics, as well as in some communication engineering and technology \cite{Freiling,Kostrykin,Kuchment,Novikov,Pokornyi,Shieh}. The Sturm-Liouville operators on the graphs actually model the wave propagation and the diffusion process transmitting within the certain networks. It is all about sending or receiving signals back and forth within the network whether there are nodes or without nodes. This problem is driven both from the theoretical and commercial requirements to solve some specific problems related to telecommunication, electric power distribution, transportation route design, and signals transferring in all sorts of networks, with regulators or without regulators. The non-local boundary condition is a critical aspect that concerns the soundness and the robustness of modeling. In most of the cases, we require the network to be capable of two-way communications and following the physical laws of conservation of energy. The global law of conservation law is decisive in most of physical systems, it turns out that the non-local boundary condition ingenuously plays a role.

\par
The most interesting aspect of such modeling is its connection to the theory of quantum graphs, which studies the physical dynamics on metric graphs governed by various sorts of differential operators \cite{Alb,Bon2,Bon,Bru,Freiling,Gerasimenko,Kottos,Naimark}, along with many different kinds of boundary conditions which often compromise the energy conservation laws or even certain commercial requirements in real world scenarios. The boundary condition acts as an integral part of the modeling these specific problems. It plays a role in many engineering projects. We are interested to formulate and to solve a few number of direct and inverse spectral problems and scattering problems on quantum graphs \cite{Belishev,Boman,Carlson,Kostrykin,Kurasov,Pivovarchik1,Pivovarchik2,Yurko1,Yurko2}.
Specifically, in this paper, we study the inverse spectral problems for a Schr\"{o}dinger operator with non-local potential functions on a star graph. Previously, the inverse spectral problems on a finite length interval for the Sturm-Liouville problems with non-local potential with various boundary value condition were considered in \cite{Alb,Alb2,Bon,Nizhnik1,Nizhnik2,Pivovarchik1,Yang,Shieh}.
\par
Let us consider the following star graph $T$: The central vertex of $T$ is located at the origin and connected by 
$m$ edges, each sharing the origin as their common vertex.  Each edge $e_{j}$ has length $l_{j}$, $j=1,2,\ldots,m$. 
    On each edge $e_{j}$, we assume that there exists the function $\psi_{j}(x)\in W_{2}^{2}(0,l_{j})$ such that the function $\psi_{j}(x)$ satisfies the following eigenvalue problem with complex-valued non-local potentials $q_{j}(x)\in L^{2}(0,l_{j})$:
\begin{eqnarray}\label{1.1}
-\frac{d^{2}}{dx^{2}}\psi_{j}(x)+q_{j}(x)\psi_{j}(0)=\lambda\psi_{j}(x),\,0<x<l_{j},\,j=1,2,\ldots,m,
\end{eqnarray}
in which $m\geq2$, $\lambda\in\mathbb{C}$ is the spectral parameter,
and it is imposed with the boundary conditions
\begin{equation}\label{1.2}
\psi_{j}(l_{j})=0,\,j=1,2,\ldots,m;\,\psi_{1}(0)=\psi_{2}(0)=\cdots=\psi_{m}(0),
\end{equation}
and the non-local boundary value condition
\begin{equation}\label{1.3}
\sum_{j=1}^{m}\big\{\frac{d}{dx}\psi_{j}(0)-\int_{0}^{l_{j}}\psi_{j}(x)\overline{q_{j}(x)}dx\big\}=0,
\end{equation}
in which $\lambda=z^{2}$, $z\in\mathbb{C}$, and the solution $\psi_{j}(x,z)$ depends on $(x,z)\in\mathbb{R}\times\mathbb{C}$ for each $j$.
We note that the inverse spectral problems is posed with a frozen argument at the vertex that is called a non-local point interaction in quantum mechanics, and  the equation actually reads as
\begin{equation}\label{1.6}
\sum_{j=1}^{m}\big\{\frac{d}{dx}\psi_{j}(0,z)-\int_{0}^{l_{j}}\psi_{j}(x,z)\overline{q_{j}(x)}dx\big\}=0,\,z\in\mathbb{C}.
\end{equation}

\section{Preliminaries}
Let us recall the Fourier sines series in $L^{2}(0,l_{j})$:
\begin{eqnarray}
&&q_{j}(x)=\sum_{n=1}^{\infty}q_{j,n}\sin[\frac{n\pi}{l_{j}} (l_{j}-x)];\\
&&q_{j,n}=\frac{2}{l_{j}}\int_{0}^{l_{j}}q_{j}(x)\sin[\frac{n\pi}{l_{j}} (l_{j}-x)]dx.
\end{eqnarray}
We note that $q_{j}(0)=0$ and $q_{j}(l_{j})=0$, and then consider a special solution of equation~(\ref{1.1}) with $\lambda=z^{2}$. Additionally, we want to meet condition~(\ref{1.2}). Then, we choose
\begin{equation} \label{2.2}
\phi_{j}(x;z)=\Big(\sin{z(l_{j}-x)}+\sin{z l_{j}}\sum_{n=1}^{\infty}\frac{q_{j,n}\sin[\frac{n\pi}{l_{j}}(l_{j}-x)]}{z^{2}-(\frac{n\pi}{l_{j}})^{2}}\Big)\prod_{k\neq j}\sin{z l_{k}}.
\end{equation}
Hence,
\begin{eqnarray}\label{2.3}
\phi_{j}'(0;z)=\Big(-z\cos{z l_{j}}-\sin{z l_{j}}\sum_{n=1}^{\infty}\frac{(-1)^{n}n\pi}{l_{j}}\frac{q_{j,n}}{z^{2}-(\frac{n\pi}{l_{j}})^{2}}\Big)\prod_{k\neq j}\sin{z l_{k}},\,j=1,2,\ldots,m.
\end{eqnarray}
From~(\ref{2.3}), we deduce that
\begin{eqnarray}
\sum_{j=1}^{m}\phi_{j}'(0;z)=-\sum_{j=1}^{m}\Big(z\cos{z l_{j}}+\sin{z l_{j}}\sum_{n=1}^{\infty}\frac{(-1)^{n}n\pi}{l_{j}}\frac{q_{j,n}}{z^{2}-(\frac{n\pi}{l_{j}})^{2}}\Big)\prod_{k\neq j}\sin{z l_{k}}.\label{234}
\end{eqnarray}
Next we deal with the condition~(\ref{1.6}). Thus,
\begin{eqnarray}\nonumber
&&\int_{0}^{l_{j}}\phi_{j}(x;z)\overline{q_{j}(x)}dx\\&=&\Big(\int_{0}^{l_{j}}\sin{z(l_{j}-x)}\overline{q_{j}(x)}dx+\sin{z l_{j}}\sum_{n=1}^{\infty}\frac{q_{j,n}\int_{0}^{l_{j}}\sin[\frac{n\pi}{l_{j}}(l_{j}-x)]\overline{q_{j}(x)}dx}{z^{2}-(\frac{n\pi}{l_{j}})^{2}}\Big)\prod_{k\neq j}\sin{z l_{k}},
\end{eqnarray}
and then,
\begin{eqnarray}\nonumber
&&\sum_{j=1}^{m}\int_{0}^{l_{j}}\phi_{j}(x;z)\overline{q_{j}(x)}dx\\&=&\sum_{j=1}^{m}\Big(\int_{0}^{l_{j}}\sin{z(l_{j}-x)}\overline{q_{j}(x)}dx+\sin{z l_{j}}\sum_{n=1}^{\infty}\frac{q_{j,n}\int_{0}^{l_{j}}\sin[\frac{n\pi}{l_{j}}(l_{j}-x)]\overline{q_{j}(x)}dx}{z^{2}-(\frac{n\pi}{l_{j}})^{2}}\Big)\prod_{k\neq j}\sin{z l_{k}}.\label{233}
\end{eqnarray}
Therefore, the equation~(\ref{1.6}) is reduced to be the characteristic function in this paper:
\begin{eqnarray}\nonumber
\Phi(z):=\sum_{j=1}^{m}\Big(\int_{0}^{l_{j}}\sin{z(l_{j}-x)}\overline{q_{j}(x)}dx+\sin{z l_{j}}\sum_{n=1}^{\infty}\frac{q_{j,n}\overline{q_{j,n}}}{z^{2}-(\frac{n\pi}{l_{j}})^{2}}\Big)\prod_{k\neq j}\sin{z l_{k}}\\
+\sum_{j=1}^{m}\Big(z\cos{z l_{j}}+\sum_{n=1}^{\infty}\frac{(-1)^{n}n\pi}{l_{j}}\sin{z l_{j}}\frac{q_{j,n}}{z^{2}-(\frac{n\pi}{l_{j}})^{2}}\Big)\prod_{k\neq j}\sin{z l_{k}}=0,\label{Phi}
\end{eqnarray}
in which the potential functions $$\{q_{1}(x),q_{2}(x),\ldots,q_{m}(x)\}\in L^{2}(0,l_{1})\oplus L^{2}(0,l_{2})\cdots\oplus L^{2}(0,l_{m})$$are imposed on the graph $T$. 
\par
For reader's convenience, we include the following classical lemma proved by E. C. Titchmarsh, and a modern proof can be found in \cite{Tang}.
\begin{lemma}[Titchmarsh]\label{L2.1}
Let $u\in\mathcal{E}'(\mathbb{R})$, then
$$N_{\mathcal{F}(u)}(r)= \frac{|\mbox{ch supp } u|}{\pi}\big(r + o(1)\big)$$ and $$N_{f}(r)=\sum_{|z|\leq r}\frac{1}{2\pi}\oint_{z}\frac{f'(\omega)}{f(\omega)}d\omega,$$
and the phrase $''\mbox{ch supp }u''$ means the covex hull of the effective support of $u$. Moreover, $N_{f}(r)$ is the counting function of the zeros of $f$ inside a ball of radius $r$. Let us count the zeros according to their multiplicities.
\end{lemma}
\begin{definition}
We denote the following quantity as the zero density of an entire function $f$ of finite type.
\begin{equation}
\delta(f):=\lim_{r\rightarrow\infty}\frac{N_{f}(r)}{r}.
\end{equation}
\end{definition}
\begin{lemma}\label{L2.3}
Let $f$, $g$ be two entire functions of finite type with density functions $\delta(f)$ and
$\delta(g)$ respectively. Then, 
\begin{equation}\nonumber
\delta(fg)=\delta(f)+\delta(g),
\end{equation}
and
\begin{equation}
\delta(f+g)=\max\{\delta(f),\delta(g)\},\label{220}
\end{equation}if the types of two functions are not equal.
\end{lemma}
\begin{proof}
We refer to B. Ya. Levin's books \cite[p.\,52]{Levin,Levin2} for a detailed introduction.  It is all about completely regular growth functions.

\end{proof}
\begin{lemma}\label{L2.4}
For $f(x)$ in $L^{2}(0,\pi)$, $z\in\mathbb{C}$, and $\int_{0}^{\pi}e^{-iz x}f(x)dx=\int_{0}^{\pi}\cos{z x}f(x)dx-i\int_{0}^{\pi}\sin{z x}f(x)dx$, the entire function $$\int_{0}^{\pi}e^{-iz x}f(x)dx;\\ \int_{0}^{\pi}\cos{z x}f(x)dx,\mbox{ and }\int_{0}^{\pi}\sin{z x}f(x)dx$$ have the same zero density in $\mathbb{C}$.
\end{lemma}
\begin{proof}
Let us simply apply Lemma \ref{L2.3}.

\end{proof}


\section{Results}
The first result is one kind of Weyl's spectral density asymptotics.
\begin{theorem}[Weyl's Law]\label{T31}
The zero density of $\Phi(z)$ is $\frac{\sum_{j=1}^{m}l_{j}}{\pi}$.
\end{theorem}
\begin{proof}
Let us observe the characteristics function
\begin{eqnarray}\nonumber
&&\Phi(z):=\sum_{j=1}^{m}\Big(\int_{0}^{l_{j}}\sin{z(l_{j}-x)}\overline{q_{j}(x)}dx+\sin{z l_{j}}\sum_{n=1}^{\infty}\frac{q_{j,n}\overline{q_{j,n}}}{z^{2}-(\frac{n\pi}{l_{j}})^{2}}\Big)\,\prod_{k\neq j}\sin{z l_{k}}\\
&&\hspace{57pt}+\sum_{j=1}^{m}\Big(z\cos{z l_{j}}+\sum_{n=1}^{\infty}\frac{(-1)^{n}n\pi}{l_{j}}\sin{z l_{j}}\frac{q_{j,n}}{z^{2}-(\frac{n\pi}{l_{j}})^{2}}\Big)\,\prod_{k\neq j}\sin{z l_{k}}.\label{31}
\end{eqnarray}
We use the interpolation theory in complex analysis referring to Levin \cite[p.\,150,\,151]{Levin2}, and rewrite a few phrases in~(\ref{31}):
\begin{equation}
\sin{z l_{j}}\sum_{n=1}^{\infty}\frac{(-1)^{n}n\pi}{l_{j}}\frac{q_{j,n}}{z^{2}-(\frac{n\pi}{l_{j}})^{2}}=\int_{0}^{l_{j}}\sin{z(l_{j}-x)}q_{j}(x)dx.
\end{equation}
The phrase $$\sin{z l_{j}}\sum_{n=1}^{\infty}\frac{q_{j,n}\overline{q_{j,n}}}{z^{2}-(\frac{n\pi}{l_{j}})^{2}}$$ means another Fourier transform. Using Lemma \ref{L2.1}, Lemma \ref{L2.3}, and Lemma \ref{2.3}, we look at the density of 
$$\sum_{j=1}^{m}z\cos{z l_{j}}\prod_{k\neq j}\sin{z l_{k}}.$$
That is,
\begin{equation}\nonumber
\delta(\sum_{j=1}^{m}z\cos{z l_{j}}\prod_{k\neq j}\sin{z l_{k}})=\max_{j}\{\delta(z\cos{z l_{j}}\prod_{k\neq j}\sin{z l_{k}})\}=\frac{\sum_{j=1}^{m}l_{j}}{\pi}.
\end{equation}
\end{proof}
\begin{definition}
We say real numbers $\{a_{1},\,a_{2},\ldots, a_{m}\},\,m\geq 2$, are called rationally independent if the identity $ \sum_{j=1}^{m}n_{j} a_{j}=0$,  for $n_{j}\in\mathbb{Z}$, implies that $$n_{1}=n_{2}=\cdots=n_{m}=0.$$
\end{definition}

\begin{theorem}[Inverse Uniqueness]\label{T3.2}
Let the equations~(\ref{1.1}),~(\ref{1.2}), and~(\ref{1.3}) be defined on the fixed known star graph 
$T$ with $0$ as its vertex along with $m$ edges. We assume the positive length of edges $\{l_{1},l_{2},\ldots,l_{m}\}$ are rationally independent. If there are two sets of the potential functions $$\{q_{1}^{\sigma}(x),q_{2}^{\sigma}(x),\ldots,q_{m}^{\sigma}(x)\}\in L^{2}(0,l_{1})\oplus L^{2}(0,l_{2})\cdots\oplus L^{2}(0,l_{m}),\,\sigma=1,2,$$ that imposed on the graph $T$ and sharing the identical $\Phi^{\sigma}(z),\,\sigma=1,2,$ defined as in~(\ref{Phi}), then $$q_{j}^{1}(x)\equiv q_{j}^{2}(x)$$ almost everywhere for all $j$.
\end{theorem}
\begin{proof}
Let us set $$\widetilde{q}_{j}(x):=q_{j}^{1}(x)-q_{j}^{2}(x),\,j=1,2,\cdots,m,$$
and suppose that $\widetilde{q}_{j}(x)\not\equiv 0$ for all $1\leq j\leq m$, and seek to a contradiction.
\par 
We begin with the subtraction of the $\Phi^{1}(z)$ and $\Phi^{2}(z)$ from two parities, and obtain
\begin{eqnarray}\nonumber
\sum_{j=1}^{m}\Big(\int_{0}^{l_{j}}\sin{z(l_{j}-x)}\overline{\widetilde{q}_{j}(x)}dx+\sin{z l_{j}}\sum_{n=1}^{\infty}\frac{q^{1}_{j,n}\overline{q^{1}_{j,n}}-q^{2}_{j,n}\overline{q^{2}_{j,n}}}{z^{2}-(\frac{n\pi}{l_{j}})^{2}}\Big)\prod_{k\neq j}\sin{z l_{k}}\\
+\sum_{j=1}^{m}\Big\{\sin{z l_{j}}\sum_{n=1}^{\infty}\frac{(-1)^{n}n\pi}{l_{j}}\frac{\widetilde{q}_{j,n}}{z^{2}-(\frac{n\pi}{l_{j}})^{2}}\Big\}\prod_{k\neq j}\sin{z l_{k}}\equiv0,\,z\in\mathbb{C}.\label{3.1}
\end{eqnarray}
We apply the interpolation theory \cite[p.\,150,\,151]{Levin2} to rewrite the equation~(\ref{3.1}), 
\begin{equation}
\sin{z l_{j}}\sum_{n=1}^{\infty}\frac{(-1)^{n}n\pi}{l_{j}}\frac{\widetilde{q}_{j,n}}{z^{2}-(\frac{n\pi}{l_{j}})^{2}}=\int_{0}^{l_{j}}\sin{z(l_{j}-x)}\widetilde{q}_{j}(x)dx,
\end{equation}
in which the later one has the same zero density as 
$\int_{0}^{l_{j}}\sin{z(l_{j}-x)}\overline{\widetilde{q}_{j}(x)}dx$.
Furthermore, we use Lemma \ref{2.3} to count the zero density for some major parts in~(\ref{3.1}):
\begin{eqnarray}\nonumber
&&\delta(\sum_{j=1}^{m}\Big\{\sin{z l_{j}}\sum_{n=1}^{\infty}\frac{\widetilde{q}_{j,n}}{z^{2}-(\frac{n\pi}{l_{j}})^{2}}\Big\}\prod_{k\neq j}\sin{z l_{k}})\\\nonumber
&=&\max_{j}\{\delta(\Big\{\int_{0}^{l_{j}}\sin{z(l_{j}-x)}\widetilde{q}_{j}(x)dx\Big\}\prod_{k\neq j}\sin{z l_{k}})\}\\&=&\max_{j}\{\delta(\int_{0}^{l_{j}}\sin{z(l_{j}-x)}\widetilde{q}_{j}(x)dx)+\delta(\prod_{k\neq j}\sin{z l_{k}})\}\\
&=&\delta(\int_{0}^{l_{j_{0}}}\sin{z(l_{j_{0}}-x)}\widetilde{q}_{j_{0}}(x)dx)+\frac{\sum_{k\neq j_{0}}^{m}l_{k}}{\pi}\label{3.4},
\end{eqnarray}
for some $j=j_{0}$. Moreover, we apply Lemma \ref{L2.1} to obtain that
$$\delta(\int_{0}^{l_{j}}\sin{z(l_{j}-x)}\widetilde{q}_{j}(x)dx)=\frac{|\mbox{ch supp } \widetilde{q}_{j}|}{\pi},\,j=1,2,\cdots,m.$$
Similarly, we apply~(\ref{220}) again to obtain
\begin{eqnarray}\nonumber
&&\delta(\sum_{j=1}^{m}\int_{0}^{l_{j}}\sin{z(l_{j}-x)}\overline{\widetilde{q}_{j}(x)}dx\prod_{k\neq j}\sin{z l_{k}})\\
&=&\max_{j}\{\delta(\int_{0}^{l_{j}}\sin{z(l_{j}-x)}\overline{\widetilde{q}_{j}(x)}dx\prod_{k\neq j}\sin{z l_{k}})\}\\
&=&\max_{j_{0}}\{\delta(\int_{0}^{l_{j_{0}}}\sin{z(l_{j_{0}}-x)}\overline{\widetilde{q}_{j_{0}}(x)}dx)+\frac{\sum_{k\neq j_{0}}^{m}l_{k}}{\pi}\},\label{3.6}
\end{eqnarray}
for some $j=j_{0}$.  
Using the interpolation theory again,
\begin{eqnarray}\nonumber
&&\delta(\sum_{j=1}^{m}\sin{z l_{j}}\sum_{n=1}^{\infty}\frac{q^{1}_{j,n}\overline{q^{1}_{j,n}}-q^{2}_{j,n}\overline{q^{2}_{j,n}}}{z^{2}-(\frac{n\pi}{l_{j}})^{2}}\prod_{k\neq j}\sin{z l_{k}})\\
&=&\max_{j}\{\delta(\sin{z l_{j}}\sum_{n=1}^{\infty}\frac{q^{1}_{j,n}\overline{q^{1}_{j,n}}-q^{2}_{j,n}\overline{q^{2}_{j,n}}}{z^{2}-(\frac{n\pi}{l_{j}})^{2}})+\delta(\prod_{k\neq j}\sin{z l_{k}})\}\\
&=&\delta(\sin{z l_{j_{0}}}\sum_{n=1}^{\infty}\frac{q^{1}_{j_{0},n}\overline{q^{1}_{j_{0},n}}-q^{2}_{j_{0},n}\overline{q^{2}_{j_{0},n}}}{z^{2}-(\frac{n\pi}{l_{j_{0}}})^{2}})+\frac{\sum_{k\neq j_{0}}^{m}l_{k}}{\pi},\label{3.7}
\end{eqnarray}
for some $j=j_{0}$. We seek to balance the zero densities on both sides of 
\begin{eqnarray}\nonumber
&&\sum_{j=1}^{m}\Big(\int_{0}^{l_{j}}\sin{z(l_{j}-x)}\overline{\widetilde{q}_{j}(x)}dx+\sin{z l_{j}}\sum_{n=1}^{\infty}\frac{q^{1}_{j,n}\overline{q^{1}_{j,n}}-q^{2}_{j,n}\overline{q^{2}_{j,n}}}{z^{2}-(\frac{n\pi}{l_{j}})^{2}}\Big)\prod_{k\neq j}\sin{z l_{k}}\\
&\equiv&-\sum_{j=1}^{m}\Big\{\int_{0}^{l_{j}}\sin{z(l_{j}-x)}\widetilde{q}_{j}(x)dx\Big\}\prod_{k\neq j}\sin{z l_{k}},\,z\in\mathbb{C},\label{3.11}
\end{eqnarray}
in which we note that $$\delta(\int_{0}^{l_{j}}\sin{z(l_{j}-x)}\overline{\widetilde{q}_{j}(x)}dx)=\delta(\int_{0}^{l_{j}}\sin{z(l_{j}-x)}\widetilde{q}_{j}(x)dx)=\frac{|\mbox{ch supp } \widetilde{q}_{j}|}{\pi},\,j=1,2,\cdots,m.$$
Moreover, using the interpolation theory again, for some $G_{j}(x)\in L^{2}(0,l_{j})$,
\begin{equation}
\sin{z l_{j}}\sum_{n=1}^{\infty}\frac{q^{1}_{j,n}\overline{q^{1}_{j,n}}-q^{2}_{j,n}\overline{q^{2}_{j,n}}}{z^{2}-(\frac{n\pi}{l_{j}})^{2}}=\int_{0}^{l_{j}}\sin{z(l_{j}-x)}G_{j}(x)dx=o(e^{|\mbox{ch supp } \widetilde{q}_{j}||\Im z|}), \mbox{ for large }z.\label{312}
\end{equation}
Most importantly, we combine the zero set of $\int_{0}^{l_{j}}\sin{z(l_{j}-x)}\overline{\widetilde{q}_{j}(x)}dx$ and the one of $\sin{z l_{j}}$ together  as
\begin{equation}
\mathcal{Z}_{j},\,j=1,2,\ldots,m.
\end{equation}
We recall the rational independence in theorem assumption and plug all of the elements of $\mathcal{Z}_{j}$ consecutively into the equation~(\ref{3.11})
\begin{eqnarray}\nonumber
&&\hspace{8pt}\Big(\int_{0}^{l_{1}}\sin{z(l_{1}-x)}\overline{\widetilde{q}_{1}(x)}dx+\,\sin{z l_{1}}\,\sum_{n=1}^{\infty}\frac{q^{1}_{1,n}\overline{q^{1}_{1,n}}-q^{2}_{1,n}\overline{q^{2}_{1,n}}}{z^{2}-(\frac{n\pi}{l_{1}})^{2}}\Big)\,\prod_{k\neq 1}\sin{z l_{k}}\\\nonumber
&&+\Big(\int_{0}^{l_{2}}\sin{z(l_{2}-x)}\overline{\widetilde{q}_{1}(x)}dx+\,\sin{z l_{2}}\,\sum_{n=1}^{\infty}\frac{q^{1}_{1,n}\overline{q^{1}_{1,n}}-q^{2}_{1,n}\overline{q^{2}_{1,n}}}{z^{2}-(\frac{n\pi}{l_{2}})^{2}}\Big)\,\prod_{k\neq 2}\sin{z l_{k}}\\\nonumber
&&\hspace{80pt}\vdots\\\nonumber\hspace{60pt}
&&+\Big(\int_{0}^{l_{m}}\sin{z(l_{m}-x)}\overline{\widetilde{q}_{1}(x)}dx\hspace{-2pt}+\sin{z l_{m}}\sum_{n=1}^{\infty}\frac{q^{1}_{1,n}\overline{q^{1}_{1,n}}-q^{2}_{1,n}\overline{q^{2}_{1,n}}}{z^{2}-(\frac{n\pi}{l_{m}})^{2}}\Big)\hspace{-3pt}\prod_{k\neq m}\sin{z l_{m}}\\
&\equiv&-\sum_{j=1}^{m}\Big\{\int_{0}^{l_{j}}\sin{z(l_{j}-x)}\widetilde{q}_{j}(x)dx\Big\}\prod_{k\neq j}\sin{z l_{k}},\,z\in\mathbb{C},\label{rational}
\end{eqnarray}
and obtain consecutively for each $1\leq j\leq m$,
\begin{eqnarray}\nonumber
\int_{0}^{l_{j}}\sin{z(l_{j}-x)}G_{j}(x)dx=0,\,j=1,2,\ldots,m,\,z\in\mathcal{Z}_{j}. 
\end{eqnarray}
Considering the zero density of $\mathcal{Z}_{j}$, which is  $$\delta(\int_{0}^{l_{j}}\sin{z(l_{j}-x)}G_{j}(x)dx)=\frac{|\mbox{ch supp } \widetilde{q}_{j}|}{\pi}+\frac{\sum_{k\neq j}l_{k}}{\pi}.$$ 
This contradicts to Lemma \ref{L2.1} and~(\ref{312}). Hence,
\begin{eqnarray}\nonumber
\int_{0}^{l_{j}}\sin{z(l_{j}-x)}G_{j}(x)dx\equiv_{\mathbb{C}}0,\,j=1,2,\ldots,m,
\end{eqnarray}
and then, 
$$q^{1}_{j,n}\overline{q^{1}_{j,n}}=q^{2}_{j,n}\overline{q^{2}_{j,n}}\,\,,$$ for all $j,n$. Using this, we go back to the equation~(\ref{3.11}), and deduce that
\begin{eqnarray}\nonumber
&&\{\int_{0}^{l_{1}}\sin{z(l_{1}-x)}\widetilde{q}_{1}(x)dx+\int_{0}^{l_{1}}\sin{z(l_{1}-x)}\overline{\widetilde{q}_{1}(x)}dx\}\prod_{k\neq 1}\sin{z l_{k}}\\\nonumber
&+&\{\int_{0}^{l_{2}}\sin{z(l_{2}-x)}\widetilde{q}_{2}(x)dx+\int_{0}^{l_{2}}\sin{z(l_{2}-x)}\overline{\widetilde{q}_{2}(x)}dx\}\prod_{k\neq 2}\sin{z l_{k}}\\\nonumber
&&\hspace{60pt}\vdots\\\hspace{50pt}
&+&\{\int_{0}^{l_{m}}\sin{z(l_{m}-x)}\widetilde{q}_{m}(x)dx+\int_{0}^{l_{m}}\sin{z(l_{m}-x)}\overline{\widetilde{q}_{m}(x)}dx\}\prod_{k\neq m}\sin{z l_{k}}\equiv_{\mathbb{C}}0.
\end{eqnarray}
Plugging the zeros from $\prod_{k\neq j}\sin{z l_{k}},\,j=1,2,\ldots,m,$ consecutively and then using Lemma \ref{L2.1} again, we find   
\begin{equation}
\int_{0}^{l_{j}}\sin{z(l_{j}-x)}\widetilde{q}_{j}(x)dx+\overline{\int_{0}^{l_{j}}\sin{z(l_{j}-x)}\widetilde{q}_{j}(x)dx}\equiv_{\mathbb{R}}0.
\end{equation}
Hence, we obtain $$\Re\{\int_{0}^{l_{j}}\sin{z(l_{j}-x)}\widetilde{q}_{j}(x)dx\}\equiv_{\mathbb{R}}0,$$
which implies that 
$$\Re\{\int_{0}^{l_{j}}\sin{z(l_{j}-x)}\widetilde{q}_{j}(x)dx\}\equiv_{\mathbb{C}}0,$$ 
and that is an entire functions of exponential type with no real part. Using the Cauchy-Riemann equation, we find that $\int_{0}^{l_{j}}\sin{z(l_{j}-x)}\widetilde{q}_{j}(x)dx$ is a constant. Let $z=0$. The constant is zero, so we obtain that $q^{1}_{j}(x)\equiv q^{2}_{j}(x)$ almost everywhere for all $j$'s. The theorem is thus proven.

\end{proof}
We may get rid of the rational independence assumption in Theorem \ref{T3.2} by adding more information about the potential functions around the vertex. Such assumptions are realistic in engineering models, since the authorities tend to know more regulator information near the control center.
\begin{theorem}[Partial Information]
Let the equations~(\ref{1.1}),~(\ref{1.2}), and~(\ref{1.3}) be defined on the fixed known star graph 
$T$ with $0$ as its vertex along with $m$ edges. We assume the convex hull of the support of $q_{j}(x)$ are strictly lesser than the length of edge 
$e_{j}$, $l_{j}$, $1\leq j\leq m$. If there are two sets of the potential functions $$\{q_{1}^{\sigma}(x),q_{2}^{\sigma}(x),\ldots,q_{m}^{\sigma}(x)\}\in L^{2}(0,l_{1})\oplus L^{2}(0,l_{2})\cdots\oplus L^{2}(0,l_{m}),\,\sigma=1,2,$$ that are imposed on the graph $T$ and sharing the identical $\Phi^{\sigma}(z),\,\sigma=1,2,$ constructed as in~(\ref{Phi}), then $q_{j}^{1}(x)\equiv q_{j}^{2}(x)$ almost everywhere for $1\leq j\leq m$.
\end{theorem}
\begin{proof}
Suppose we have two such set of potential functions  and assume that $\Phi^{1}(z)\equiv_{\mathbb{C}}\Phi^{2}(z)$.  This implies that they have the same zero density as given by Theorem \ref{T31}, which has the quantity
 $$\frac{\sum_{j=1}^{m}l_{j}}{\pi}.$$
Let us try to balance the zero densities on both sides of the equation~(\ref{3.11}), and  plug all of these common zeros into the following function
\begin{eqnarray}\nonumber
&&\sum_{j=1}^{m}\Big(\int_{0}^{l_{j}}\sin{z(l_{j}-x)}\overline{\widetilde{q}_{j}(x)}dx+\int_{0}^{l_{j}}\sin{z(l_{j}-x)}\widetilde{q}_{j}(x)dx+\int_{0}^{l_{j}}\sin{z(l_{j}-x)}G_{j}(x)dx\Big) \prod_{k\neq j}\sin{z l_{k}}.
\end{eqnarray}
Since for all $1\leq j\leq m$ by the theorem assumption,
\begin{eqnarray}
&&\delta(\int_{0}^{l_{j}}\sin{z(l_{j}-x)}\overline{\widetilde{q}_{j}(x)}dx)\,<\frac{l_{j}}{\pi};\\
&&\delta(\int_{0}^{l_{j}}\sin{z(l_{j}-x)}\widetilde{q}_{j}(x)dx)\,<\frac{l_{j}}{\pi};\\
&&\delta(\int_{0}^{l_{j}}\sin{z(l_{j}-x)}G_{j}(x)dx)<\frac{l_{j}}{\pi}.
\end{eqnarray}
Let us apply~(\ref{3.4}),~(\ref{3.6}), and~(\ref{3.7}) to reach the contradiction. The theorem is thus proven.

\end{proof}

\begin{acknowledgement}
The author sincerely  thanks Professor Chung-Tsun Shieh at Tamkang University, and Professor Natalia P. Bondarenko at Peoples' Friendship University of Russia (RUDN University) for many useful advises during the preparation of this manuscript. The author also ingenuously appreciates the research funding supported by NSTC under the project number 113-2115-M-131 -001. The content of this manuscript does not necessarily reflect the position or the policy of administration, and no official endorsement should be inferred, neither.

\end{acknowledgement}

\end{document}